\documentclass[11pt]{article}
\usepackage{latexsym,amsthm,amsmath,amssymb}
\usepackage[utf8x]{inputenc}

\newtheorem{theorem}{Theorem}
\newtheorem{lemma}{Lemma}

\usepackage{todonotes}
\usepackage{subcaption}
\usepackage{tikz}
\usepackage{makecell}
\usepackage{url}
\usepackage{hyperref}
\usepackage{enumitem}

\tikzset{bul/.style={circle,fill=black,inner sep=1pt,draw=black}}
\tikzset{lab/.style={sloped,anchor=south,auto=false,inner sep=2pt,black}}
\tikzset{>=latex}

\newtheorem{observation}{Observation}

\begin{document}

\begin{center}
{\Large\textbf{Planar Steiner Orientation is NP-complete}}

\vspace{2ex}

\textbf{Moritz Beck, Johannes Blum, Myroslav Kryven,\\ Andre Löffler, and Johannes Zink}\\
Julius-Maximilians-Universit\"at W\"urzburg, Germany\\
\texttt{\{beck,blum,kryven,loeffler,zink\}@informatik.uni-wuerzburg.de}
\end{center}

\begin{abstract}
Many applications in graph theory are motivated by routing or flow problems.
Among these problems is \textsc{Steiner Orientation}:
given a mixed graph $G$ (having directed and undirected edges) and a set $T$ of $k$ terminal pairs in $G$,
is there an orientation of the undirected edges in~$G$ such that there is a directed path for every terminal pair in $T$?
This problem was shown to be $\mathit{NP}$-complete by Arkin and Hassin~\cite{arkin02}
and later $W[1]$-hard by Pilipczuk and Wahlstr\"{o}m~\cite{pilipczuk16}, parametrized by $k$.
On the other hand, there is an $\mathit{XP}$ algorithm by Cygan et al.~\cite{cygan13}
and a polynomial time algorithm for graphs without directed edges by Hassin and Megiddo~\cite{hassin89}.
Chitnis and Feldmann~\cite{feldmann18} showed
$W[1]$-hardness of the problem for
graphs of genus~1.

We consider a further restriction to planar
graphs and show $\mathit{NP}$-completeness.
\end{abstract}

\section{Introduction}
Consider the following routing problem on mixed graphs:

\noindent\begin{tabular}{|lp{.85\textwidth}|}
 \hline
&\textsc{Steiner Orientation}\\
\textbf{Input:} &
Mixed graph $G=(V,E\cup A)$ with undirected edges $E$ and directed arcs $A$,\\
&set $T \subseteq V \times V$ of $k$ terminal pairs\\
\textbf{Output:} &
Orientation of all edges in $E$, such that for every $(s,t) \in T$ there is an $s$-$t$-path in $G$\\
 \hline
\end{tabular}

First notice that there is a polynomial time algorithm by Hassin and Megiddo~\cite{hassin89}
for the case where $G$ has no directed edges ($A = \emptyset$).
On general graphs, this problem is in fact $\mathit{NP}$-complete, shown  by Arkin and Hassin~\cite{arkin02},
also providing an efficient algorithm for the special case of $k=2$.
Generalizing this result to $k\geq 1$, Cygan et al.~\cite{cygan13} gave a $n^{O(k)}$ time algorithm.
This means that \textsc{Steiner Orientation} is in $\mathit{XP}$, parametrized by the number of terminal pairs.
This raised the question if \textsc{Steiner Orientation} is \emph{fixed-parameter tractable}:
is there an algorithm with runtime $f(k)\cdot n^{O(1)}$ for some computable function $f$ only dependent on $k$?
Pilipczuk and Wahlstr\"{o}m~\cite{pilipczuk16} showed \textsc{Steiner Orientation} to be $W[1]$-hard in $k$,
disproving the existance of such an algorithm under common assumptions.

Considering $k$-\textsc{SAT}, Impagliazzo and
Paturi~\cite{Impagliazzo2001} introduced the \emph{Exponential Time
Hypothesis} (ETH):
the classic \textsc{SAT}-problem parametrized by the number of variables per clause does not have a subexponential algorithm.
When considering parametrized problems, ETH is often assumed in order to obtain bounds on runtimes.
Assuming ETH, Pilipczuk and Wahlstr\"{o}m~\cite{pilipczuk16} could also show that there is no $f(k)\cdot n^{o(k/\log k)}$
time algorithm for any computable $f$, showing that the $\textit{XP}$ algorithm by Cygan et al.~\cite{cygan13} is almost optimal.
Restricting the problem, Chitnis and Feldmann~\cite{feldmann18} showed that the problem remains $W[1]$-hard
and the $\mathit{XP}$ algorithm almost optimal even if the input graph has genus~1.
All hardness-proofs provided utilize non-planar instances. This leaves open the following question:
What is the computational complexity of \textsc{Steiner Orientation} on planar graphs?

In this work, we consider the \textsc{Planar Steiner Orientation} problem where $G$ is a planar graph.
As a first result on computational complexity, we show the following:

\begin{theorem}
\label{thm:theorem}
\textsc{Planar Steiner Orientation} is $\mathit{NP}$-complete.
\end{theorem}

\section{Hardness Proof}
To prove Theorem~\ref{thm:theorem}, we give a reduction from \textsc{Planar Monotone 3-SAT},
introduced by de~Berg et al. \cite{berg12} and known to be $\mathit{NP}$-complete.
We use different gadgets for variables, clauses and edges.
These are stitched together at shared undirected edges.
Given a planar monotone 3-SAT formula~$F$, we use these gadgets to create an instance of
\textsc{Planar Steiner Orientation} resembling the incidence graph of $F$ with $\lvert T \rvert$ polynomial in the size of $F$.
Without loss of generality we assume that every variable of $F$ occurs both negated and unnegated.

Figure~\ref{fig:flip_gadget} shows a \emph{flip gadget}, a building block used in other gadgets.
It contains two terminal pairs $(s_1, t_1)$ and $(s_2, t_2)$ and two undirected (red) edges.
Connecting both pairs will result in opposing directions for the two undirected edges.

For every variable $x$ in $F$, we have a \emph{variable gadget} (see Figure~\ref{fig:variable_gadget}).
It mimics the flip gadget, providing an undirected edge $e^x_C$ for every positive/negative clause $C$ containing $x$ above/below the terminal pairs respectively.
We say that the gadget is (false) true if the undirected edges are oriented (counter-)clockwise.
		No other orientation allows connecting both terminal pairs.

For every clause $C$, we have a \emph{clause gadget} (see Figure~\ref{fig:clause_gadget}).
It contains a terminal pair $(s, t)$ and has an undirected edge $\overline{ e}^w_C$ for each variable $w$ it contains.
The undirected edge $\overline{ e}^y_C$ in the middle is flipped to get a consistent orientation for variables set to true.
The edges $f$ and $g$ are synchronized by two flip gadgets to ensure that at most one of them is used to connect $(s,t)$.
For clauses with only two variables we simply replace the edge $\tilde{e}^y_C$
with an arc from left to right and omit the attached flip gadget
(see Figure~\ref{fig:full-example}, clause $(X\vee Y)$).
It is easy to see that this way no new possibilities for $s$-$t$-paths are created,
keeping the gadget valid.

We use two stacked flip gadgets as \emph{edge gadgets} resembling the variable-clause-incidences.
By reversing the direction twice, we synchronize the two red edges $e^x_C$ and $\overline{ e}^x_C$ for all $x$ and $C$.

An important property of our construction is that all the gadgets that we use are \emph{self-contained}, which means that for each gadget any simple path connecting an $(s, t)$-pair of the gadget stays inside the gadget.
We state the following simple observation:

\begin{observation}
	Every source $s$ has indegree zero and every target $t$ has outdegree zero.\label{obs:indeg}
\end{observation}

We use this to prove the following lemma.

\begin{lemma}
	\label{lem:self-containment}
	In our construction each clause, edge, and variable gadget is self-contained.
\end{lemma}
\begin{proof}
	\textbf{Clause gadgets.} Assume there is a simple path connecting a terminal pair of a clause gadget, which is not fully contained within the gadget. Then there must be an edge that leaves the clause gadget and due to the structure, this edge must be part of an edge gadget. But all edges leaving a clause gadget and entering an edge gadget end in some target terminal, so by Observation~\ref{obs:indeg} the path cannot re-enter the clause gadget.

	\textbf{Edge gadgets.} Consider a simple path that leaves an edge gadget. If the leaving edge is part of a clause gadget, the path leads to a target terminal within the clause gadget or within another edge gadget, from where it cannot re-enter the original edge gadget.  The case where the leaving edge is part of a variable gadget is similar.

	\textbf{Variable gadgets.} Consider a simple path that leaves a variable gadget. Then the leaving edge is part of an edge gadget and leads to a target terminal, so the path cannot re-enter the variable gadget.\qedhere

\end{proof}

\begin{lemma}
		All terminal pairs of a clause gadget corresponding to a clause $C$ can be connected if and only if at least one of the edges in $\{ \overline{e}^x_C, \overline{e}^y_C, \overline{e}^z_C\}$ is directed to the~right.
\end{lemma}
\begin{proof}
	According to Lemma~\ref{lem:self-containment} it suffices to consider only paths within the clause gadget as there is a path connecting a terminal pair if and only if there is a simple path connecting this terminal pair.
	The terminal pairs of the flip gadgets can be connected within the flip gadgets.
	They cannot be connected otherwise within the clause gadget since they would have to pass through a source or target vertex then, which contradicts Observation~\ref{obs:indeg}.
	The edges $f$ and $g$ are both directed upwards or downwards because of the two flip gadgets between them.
	Hence it suffices to show the equivalence for the pair $(s,t)$.

	``$\Leftarrow$'':
	Case 1: If $\overline e^x_C$ is directed to the right, orient the edge $f$ away from $s$.
	Case 2: If $\overline e^y_C$ is directed to the right, $\tilde{e}^y_C$ is directed to the left.
	Case 3: If $\overline e^z_C$ is directed to the right, orient the edge $g$ pointing to $t$.
	In each of these cases $s$ is connected to $t$.

	``$\Rightarrow$'': By contraposition.
	As we move away from $s$, we can neither use the edge $\tilde{e}^y_C$ nor $\overline e^z_C$.
	Thus we have to use $f$ which means it points away from $s$.
	To come to $t$ we have to use one of the edges $\overline e^x_C$ or $g$.
	This is impossible.
\end{proof}

The next two lemmas follow immediately from the gadget structure and Lemma~\ref{lem:self-containment}.

\begin{lemma}
The only way to connect all terminal pairs of a variable gadget is to orient all edges clockwise or counterclockwise.
\end{lemma}

\begin{lemma}
The only way to connect all terminal pairs of an edge gadget is to orient the two outer edges $e^x_C$ and $\overline{e}^x_C$ in the same direction. 
\end{lemma}

We can observe that there is always an orientation such that the terminal pairs of all variable and edge gadgets can be connected.
Provided that all terminal pairs of all variable and edge gadgets are connected, the terminal pairs of all clause gadgets can be connected if and only if the formula is satisfiable.
Thus, the \textsc{Planar Steiner Orientation} instance has a solution if and only if the corresponding \textsc{Planar Monotone 3-SAT} formula is satisfiable.
This proves Theorem~\ref{thm:theorem}.

\begin{figure}[tb]
 \begin{minipage}{\textwidth}
    \begin{subfigure}[b]{0.2\textwidth}
      \centering
      \captionsetup{justification=centering}
\begin{tikzpicture}[scale=.4]
 \draw (0,0)node[bul,label=below:{$t_2$}](t2){} (1,0)node[bul,label=right:{$s_1$}](s1){} (4,0)node[bul,label=left:{$s_2$}](s2){} (5,0)node[bul,label=above:{$t_1$}](t1){};
 \draw (1.5,-1.5)node[bul](ll){} (3.5,-1.5)node[bul](lr){} (1.5,1.5)node[bul](ul){} (3.5,1.5)node[bul](ur){};
 \draw [->](s1) to (ul);\draw [->] (s1) to (ll);
 \draw [->](s2) to (ur);\draw [->] (s2) to (lr);
 \draw [->](ul) to (t2);\draw [->] (ll) to (t2);
 \draw [->](ur) to (t1);\draw [->] (lr) to (t1);
 \draw [red,line width=2pt] (ul) -- (ur) (ll) -- (lr);
 \draw (0,-3.7)node[bul,draw=white,fill=white](){};
\end{tikzpicture}

       \caption{Flip gadget}
      \label{fig:flip_gadget}
    \end{subfigure}
    \hfill
    \begin{subfigure}[b]{0.35\textwidth}
      \centering
      \captionsetup{justification=centering}
      \begin{tikzpicture}[scale=.4]
 \draw (0,0)node[bul,label=below:{$t_2$}](t2){} (1,0)node[bul,label=right:{$s_1$}](s1){} (11,0)node[bul,label=left:{$s_2$}](s2){} (12,0)node[bul,label=above:{$t_1$}](t1){};
 \draw (1.5,-1.5)node[bul](ll){} (1.5,1.5)node[bul](ul){}
 (3.5,-1.5)node[bul](lm1){} (3.5,1.5)node[bul](um1){}
 (5,1.5)node[bul](um21){} (7,1.5)node[bul](um22){}
 (8.5,-1.5)node[bul](lm3){} (8.5,1.5)node[bul](um3){}
 (10.5,-1.5)node[bul](lr){} (10.5,1.5)node[bul](ur){};
 \draw [->](s1) to (ul);\draw [->] (s1) to (ll);
 \draw [->](s2) to (ur);\draw [->] (s2) to (lr);
 \draw [->](ul) to (t2);\draw [->] (ll) to (t2);
 \draw [->](ur) to (t1);\draw [->] (lr) to (t1);
 \draw [->,bend right=15](lm1) to (lm3); \draw [->,bend right=15](lm3) to (lm1);
 \draw [->,bend right=15](um1) to (um21);\draw [->,bend right=15](um21) to (um1);
 \draw [->,bend right=15](um22) to (um3);\draw [->,bend right=15](um3) to (um22);
 \draw [red, line width=2pt] (ul)--node[lab]{$e^x_{A}$} (um1) (um21)--node[lab]{$e^x_{B}$}(um22) (um3)--node[lab]{$e^x_{C}$}(ur)
 (ll)--node[lab,anchor=north]{$e^x_{D}$}(lm1) (lm3)--node[lab,anchor=north]{$e^x_{E}$}(lr);
\end{tikzpicture}
      \caption{Variable gadget}
      \label{fig:variable_gadget}
    \end{subfigure}
    \hfill
    \begin{subfigure}[b]{0.4\textwidth}
      \centering
      \captionsetup{justification=centering}
      \begin{tikzpicture}[scale=.4]
        \draw (3,2)node[bul,label=left:{$t$}](t){} (10,2)node[bul,label=right:{$s$}](s){};
       \draw (0,0)node[bul](lx){} (2,0)node[bul](rx){};
        \draw (4,1)node[bul,fill=blue!25](tt2){} (5,1)node[bul,fill=blue](ss1){} (8,1)node[bul,fill=blue!25](ss2){} (9,1)node[bul,fill=blue](tt1){}; 
        \draw (5.5,0)node[bul](ll){} (7.5,0)node[bul](lr){} (5.5,2)node[bul](ul){} (7.5,2)node[bul](ur){}; 
        \draw (11,0)node[bul](lz){} (13,0)node[bul](rz){};
        \draw (3,3.5)node[bul](t1){} (10,3.5)node[bul](s1){};
        \draw (3,5.5)node[bul](t2){} (10,5.5)node[bul](s2){};
        \draw (2,6.5)node[bul](t3){} (11,6.5)node[bul](s3){};
        \draw (6.5,4)node[bul](c1){} (6.5,6)node[bul](c2){};
        \draw (4.75,3.5)node[bul,fill=green!25](l3){} (4.75,4.5)node[bul,fill=green](ll4){} (4.75,5.3)node[bul,fill=green](l4){} (4.75,6)node[bul,fill=green!25](ll3){};
        \draw (8.25,3.5)node[bul,fill=red!25](r3){} (8.25,4.5)node[bul,fill=red](rr4){} (8.25,5.3)node[bul,fill=red](r4){} (8.25,6)node[bul,fill=red!25](rr3){};

        \draw [->] (t1)to(t);\draw [->] (t3)to(t2);\draw [->] (s)to(s1);\draw [->] (s2)to(s3);
        \draw [red,line width=2pt] (c1)--(c2) (t2)--node[color=black,anchor=east]{$g$}(t1) (s2)--node[color=black,anchor=west]{$f$}(s1);

        \draw [->] (rz)to(s3); \draw [->] (s3)to(t3); \draw [->] (t3)to(lx);
        \draw [->] (ll4)to(t1);\draw [->] (ll4)to(c1); \draw [->] (ll3)to(t2);\draw [->] (ll3)to(c2);
        \draw [->] (t1)to(l3); \draw [->] (c1)to(l3); \draw [->] (t2)to(l4);\draw [->] (c2)to(l4);

        \draw [->] (rr4)to(s1);\draw [->] (rr4)to(c1); \draw [->] (rr3)to(s2);\draw [->] (rr3)to(c2);
        \draw [->] (s1)to(r3); \draw [->] (c1)to(r3); \draw [->] (s2)to(r4);\draw [->] (c2)to(r4);
        \draw [->](ss1) to (ul);\draw [->] (ss1) to (ll);\draw [->](ss2) to (ur);\draw [->] (ss2) to (lr);
        \draw [->](ul) to (tt2);\draw [->] (ll) to (tt2);\draw [->](ur) to (tt1);\draw [->] (lr) to (tt1);
        \draw [->] (rx)to(t);\draw [->] (ul)to(t);\draw [->] (s)to(ur);\draw [->] (s)to(lz);

        \draw [red,line width=2pt] (ul)--node[lab,color=black]{$\tilde{e}^y_{C}$} (ur) (ll)--node[lab,color=black]{$\overline{ e}^y_{C}$} (lr);
        \draw [red,line width=2pt] (lx)--node[lab,color=black]{$\overline{ e}^x_{C}$} (rx) (lz)--node[lab,color=black]{$\overline{ e}^z_{C}$} (rz);
      \end{tikzpicture}
       \caption{Clause gadget}
      \label{fig:clause_gadget}
    \end{subfigure}
	\caption{(a) A flip gadget, used to construct edge gadgets, (b) a variable gadget with three positive and two negative occurrences and (c) a positive clause gadget (unlabeled $(s,t)$-pairs color-coded).}
    \label{fig:construction}
   \end{minipage}
  \end{figure}
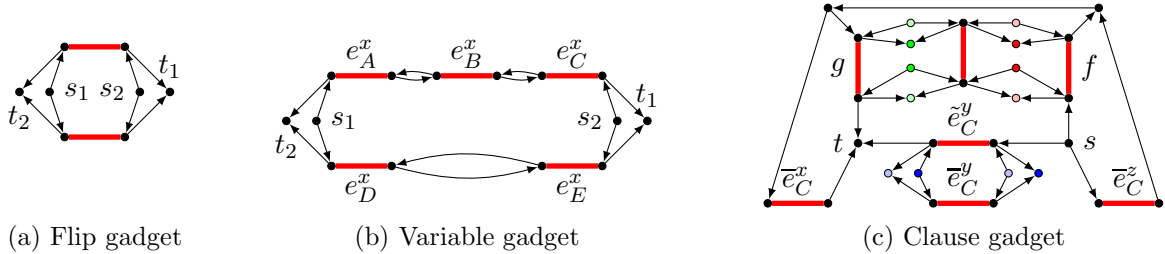

\section{Conclusion}
By a polynomial-time reduction from \textsc{Planar Monotone 3-SAT} we have shown that \textsc{Planar Steiner Orientation} is $\mathit{NP}$-hard.
Clearly, it also is in $\mathit{NP}$, which makes it $\mathit{NP}$-complete.

Future work involves proving $W[1]$-hardness or looking for approximation algorithms, connecting as many pairs as possible.
Graph classes with other restrictions could also be considered.

\paragraph{Acknowledgements}
We thank Andreas Feldmann for pointing us to this problem and we also thank the workshop HOMONOLO 2017 for the fruitful discussions.

\bibliography{bibliography}

\begin{thebibliography}{1}

\bibitem{arkin02}
Esther~M. Arkin and Refael Hassin.
\newblock A note on orientations of mixed graphs.
\newblock {\em Discrete Applied Mathematics}, 116(3):271--278, 2002.

\bibitem{feldmann18}
Rajesh Chitnis and Andreas~Emil Feldmann.
\newblock A {T}ight {L}ower {B}ound for {S}teiner {O}rientation.
\newblock In {\em Proceedings of the 13th International Computer Science
  Symposium in Russia}, 2018.
\newblock (to appear).

\bibitem{cygan13}
Marek Cygan, Guy Kortsarz, and Zeev Nutov.
\newblock {S}teiner {F}orest {O}rientation {P}roblems.
\newblock {\em {SIAM} Journal on Discrete Mathematics}, 27(3):1503--1513, 2013.

\bibitem{berg12}
Mark de~Berg and Amirali Khosravi.
\newblock Optimal binary space partitions for segments in the plane.
\newblock {\em International Journal of Computational Geometry \&
  Applications}, 22(3):187--206, 2012.

\bibitem{hassin89}
Refael Hassin and Nimrod Megiddo.
\newblock On orientations and shortest paths.
\newblock {\em Linear Algebra and its Applications}, 114:589--602, 1989.

\bibitem{Impagliazzo2001}
Russell Impagliazzo and Ramamohan Paturi.
\newblock On the {C}omplexity of k-{SAT}.
\newblock {\em Journal of Computer and System Sciences}, 62(2):367--375, 2001.

\bibitem{pilipczuk16}
Marcin Pilipczuk and Magnus Wahlstr\"{o}m.
\newblock {D}irected multicut is {W}[1]-hard, even for four terminal pairs.
\newblock In {\em Proceedings of the 27th ACM-SIAM Symposium on Discrete
  Algorithms}, pages 1167--1178, 2016.

\end{thebibliography}

\clearpage
\appendix
\section{Full Example}
We want to provide a small but complete example.
In Figure~\ref{fig:full-example} we give the incidence graph and the
\textsc{Planar Steiner Orientation} instance created using the gadgets introduced above corresponding to:
$ F = (X \vee Y) \wedge (\neg X \vee \neg Z \vee \neg W) \wedge (Y \vee Z \vee W) \wedge (\neg X \vee \neg Y \vee \neg Z) $.

\begin{figure}[h]
    \begin{subfigure}[b]{0.25\textwidth}
      \centering
        \begin{tikzpicture}
         \draw
         (0,0)node[circle,draw,inner sep=0pt,minimum width=.6cm,fill=green!50](X){$X$}
         (1,0)node[circle,draw,inner sep=0pt,minimum width=.6cm,fill=green!50](Y){$Y$}
         (2,0)node[circle,draw,inner sep=0pt,minimum width=.6cm,fill=green!50](Z){$Z$}
         (3,0)node[circle,draw,inner sep=0pt,minimum width=.6cm,fill=green!50](W){$W$};
         \draw (.4,1)node[rectangle,draw,inner sep=1pt,minimum width=1.2cm,fill=orange!50](XY){$X\vee Y$} (2.1,1)node[rectangle,draw,inner sep=1pt,minimum width=1.2cm,fill=orange!50](YZW){$Y\vee Z\vee W$};
         \draw (X)--(XY)--(Y) (Y)--(YZW)--(Z) (YZW)--(W);
         \draw (1,-1.5)node[rectangle,draw,inner sep=1pt,minimum width=1.2cm,fill=orange!50](XYZ){\makecell[c]{$\neg X\vee \neg Y$ \\ $\vee \neg Z$}} (1.6,-3)node[rectangle,draw,inner sep=1pt,minimum width=1.2cm,fill=orange!50](XZW){\makecell[c]{$\neg X\vee \neg Z$\\$\vee \neg W$}};
         \draw [bend left=45] (XZW)to(X);
         \draw [bend right=15] (XZW)to(Z);
         \draw (XZW)--(W) (X)--(XYZ)--(Y) (XYZ)--(Z);
        \end{tikzpicture}
		\caption{Incidence graph of $F$}
		\label{fig:full-example-incidence-graph}
    \end{subfigure}
    \hfill
    \begin{subfigure}[b]{0.7\textwidth}
      \centering
      \clearpage{}\begin{tikzpicture}[scale=0.2]
\begin{scope}[shift={(6.5,3.5)}] 
        \fill[orange!50] (-2,0)--(21.5,0)--(21.5,7)--(-2,7)--cycle;
         \draw (3,2)node[bul](t){} (10,2)node[bul](s){} (-1.5,0)node[bul](lx){} (.5,0)node[bul](rx){};
        \draw (4,1)node[bul,fill=blue!25](tt2){} (5,1)node[bul,fill=blue](ss1){} (8,1)node[bul,fill=blue!25](ss2){} (9,1)node[bul,fill=blue](tt1){}; 
        \draw (5.5,0)node[bul](ll){} (7.5,0)node[bul](lr){} (5.5,2)node[bul](ul){} (7.5,2)node[bul](ur){}; 
        \draw (19,0)node[bul](lz){} (21,0)node[bul](rz){}
        (3,3.5)node[bul](t1){} (10,3.5)node[bul](s1){};
        \draw (3,5.5)node[bul](t2){} (10,5.5)node[bul](s2){} (2,6.5)node[bul](t3){} (11,6.5)node[bul](s3){};
        \draw (6.5,4)node[bul](c1){} (6.5,6)node[bul](c2){};
        \draw (4.75,3.5)node[bul,fill=green!25](l3){} (4.75,4.5)node[bul,fill=green](ll4){} (4.75,5.3)node[bul,fill=green](l4){} (4.75,6)node[bul,fill=green!25](ll3){};
        \draw (8.25,3.5)node[bul,fill=red!25](r3){} (8.25,4.5)node[bul,fill=red](rr4){} (8.25,5.3)node[bul,fill=red](r4){} (8.25,6)node[bul,fill=red!25](rr3){};

        \draw [->] (t1)to(t);\draw [->] (t3)to(t2);\draw [->] (s)to(s1);\draw [->] (s2)to(s3);
        \draw [red,line width=2pt] (c1)--(c2) (t2)--(t1) (s2)--(s1);

        \draw [->] (rz)to(s3); \draw [->] (s3)to(t3); \draw [->] (t3)to(lx);
        \draw [->] (ll4)to(t1);\draw [->] (ll4)to(c1); \draw [->] (ll3)to(t2);\draw [->] (ll3)to(c2);
        \draw [->] (t1)to(l3); \draw [->] (c1)to(l3); \draw [->] (t2)to(l4);\draw [->] (c2)to(l4);

        \draw [->] (rr4)to(s1);\draw [->] (rr4)to(c1); \draw [->] (rr3)to(s2);\draw [->] (rr3)to(c2);
        \draw [->] (s1)to(r3); \draw [->] (c1)to(r3); \draw [->] (s2)to(r4);\draw [->] (c2)to(r4);
        \draw [->](ss1) to (ul);\draw [->] (ss1) to (ll);\draw [->](ss2) to (ur);\draw [->] (ss2) to (lr);
        \draw [->](ul) to (tt2);\draw [->] (ll) to (tt2);\draw [->](ur) to (tt1);\draw [->] (lr) to (tt1);
        \draw [->] (rx)to(t);\draw [->] (ul)to(t);\draw [->] (s)to(ur);\draw [->] (s)to(lz);

        \draw [red,line width=2pt] (ul)-- (ur) (ll)-- (lr);
        \draw [red,line width=2pt] (lx)-- (rx) (lz)-- (rz);
\end{scope}
\begin{scope}[shift={(-15,3.5)}] 
\fill[orange!50] (-.5,0)--(16,0)--(16,7)--(-.5,7)--cycle;
         \draw (3,2)node[bul](t){} (10,2)node[bul](s){} (0,0)node[bul](lx){} (2,0)node[bul](rx){};
        \draw (5.5,2)node[bul](ul){} (7.5,2)node[bul](ur){}; 
        \draw (13.5,0)node[bul](lz){} (15.5,0)node[bul](rz){} (3,3.5)node[bul](t1){} (10,3.5)node[bul](s1){};
        \draw (3,5.5)node[bul](t2){} (10,5.5)node[bul](s2){} (2,6.5)node[bul](t3){} (11,6.5)node[bul](s3){};
        \draw (6.5,4)node[bul](c1){} (6.5,6)node[bul](c2){};
        \draw (4.75,3.5)node[bul,fill=green!25](l3){} (4.75,4.5)node[bul,fill=green](ll4){} (4.75,5.3)node[bul,fill=green](l4){} (4.75,6)node[bul,fill=green!25](ll3){};
        \draw (8.25,3.5)node[bul,fill=red!25](r3){} (8.25,4.5)node[bul,fill=red](rr4){} (8.25,5.3)node[bul,fill=red](r4){} (8.25,6)node[bul,fill=red!25](rr3){};

        \draw [->] (t1)to(t);\draw [->] (t3)to(t2);\draw [->] (s)to(s1);\draw [->] (s2)to(s3);
        \draw [red,line width=2pt] (c1)--(c2) (t2)--(t1) (s2)--(s1);

        \draw [->] (rz)to(s3); \draw [->] (s3)to(t3); \draw [->] (t3)to(lx);
        \draw [->] (ll4)to(t1);\draw [->] (ll4)to(c1); \draw [->] (ll3)to(t2);\draw [->] (ll3)to(c2);
        \draw [->] (t1)to(l3); \draw [->] (c1)to(l3); \draw [->] (t2)to(l4);\draw [->] (c2)to(l4);

        \draw [->] (rr4)to(s1);\draw [->] (rr4)to(c1); \draw [->] (rr3)to(s2);\draw [->] (rr3)to(c2);
        \draw [->] (s1)to(r3); \draw [->] (c1)to(r3); \draw [->] (s2)to(r4);\draw [->] (c2)to(r4);
        \draw [->] (rx)to(t);\draw [->] (ul)to(t);\draw [->] (s)to(ur);\draw [->] (s)to(lz);

        \draw [->] (ul)to (ur);
        \draw [red,line width=2pt] (lx)-- (rx) (lz)-- (rz);
\end{scope}
\begin{scope}[shift={(-16.5,-4)}] 
 \fill[green!50] (0,-2)--(12,-2)--(12,2)--(0,2)--cycle;
  \draw (0,0)node[bul](t2){} (1,0)node[bul](s1){} (11,0)node[bul](s2){} (12,0)node[bul](t1){};
 \draw (1.5,-1.5)node[bul](ll){} (1.5,1.5)node[bul](ul){}
 (3.5,-1.5)node[bul](lm1){}
 (8,-1.5)node[bul](lm2){}
 (3.5,1.5)node[bul](um1){}
 (10,-1.5)node[bul](lr){} (10,1.5)node[bul](ur){};
 \draw [->](s1) to (ul);\draw [->] (s1) to (ll); \draw [->](s2) to (ur);\draw [->] (s2) to (lr);
 \draw [->](ul) to (t2);\draw [->] (ll) to (t2); \draw [->](ur) to (t1);\draw [->] (lr) to (t1);
 \draw [->,bend right=15](lm1) to (lm2); \draw [->,bend right=15](lm2) to (lm1);
 \draw [->,bend right=15](um1) to (ur);\draw [->,bend right=15](ur) to (um1);
 \draw [red, line width=2pt] (ul)-- (um1) (ll)--(lm1) (lm2)--(lr);
\end{scope}
\begin{scope}[shift={(-3,-4)}] 
\fill[green!50] (0,-2)--(12,-2)--(12,2)--(0,2)--cycle;
  \draw (0,0)node[bul](t2){} (1,0)node[bul](s1){} (11,0)node[bul](s2){} (12,0)node[bul](t1){};
 \draw (1.5,-1.5)node[bul](ll){} (1.5,1.5)node[bul](ul){}
 (3.5,-1.5)node[bul](lm1){}
 (3.5,1.5)node[bul](um1){}
 (8,1.5)node[bul](um2){}
 (10,-1.5)node[bul](lr){} (10,1.5)node[bul](ur){};
 \draw [->](s1) to (ul);\draw [->] (s1) to (ll); \draw [->](s2) to (ur);\draw [->] (s2) to (lr);
 \draw [->](ul) to (t2);\draw [->] (ll) to (t2); \draw [->](ur) to (t1);\draw [->] (lr) to (t1);
 \draw [->,bend right=15](lm1) to (lr); \draw [->,bend right=15](lr) to (lm1);
 \draw [->,bend right=15](um1) to (um2);\draw [->,bend right=15](um2) to (um1);
 \draw [red, line width=2pt] (ul)-- (um1) (um2)--(ur) (ll)--(lm1);
\end{scope}
\begin{scope}[shift={(10.5,-4)}] 
\fill[green!50] (0,-2)--(12,-2)--(12,2)--(0,2)--cycle;
  \draw (0,0)node[bul](t2){} (1,0)node[bul](s1){} (11,0)node[bul](s2){} (12,0)node[bul](t1){};
 \draw (1.5,-1.5)node[bul](ll){} (1.5,1.5)node[bul](ul){}
 (3.5,-1.5)node[bul](lm1){} (8,-1.5)node[bul](lm2){} (3.5,1.5)node[bul](um1){} (10,-1.5)node[bul](lr){} (10,1.5)node[bul](ur){};
 \draw [->](s1) to (ul);\draw [->] (s1) to (ll); \draw [->](s2) to (ur);\draw [->] (s2) to (lr);
 \draw [->](ul) to (t2);\draw [->] (ll) to (t2); \draw [->](ur) to (t1);\draw [->] (lr) to (t1);
 \draw [->,bend right=15](lm1) to (lm2); \draw [->,bend right=15](lm2) to (lm1);
 \draw [->,bend right=15](um1) to (ur);\draw [->,bend right=15](ur) to (um1);
 \draw [red, line width=2pt] (ul)-- (um1) (ll)--(lm1) (lm2)--(lr);
\end{scope}
\begin{scope}[shift={(24,-4)}] 
 \fill[green!50] (0,-2)--(5,-2)--(5,2)--(0,2)--cycle;
  \draw (0,0)node[bul](t2){} (1,0)node[bul](s1){} (4,0)node[bul](s2){} (5,0)node[bul](t1){};
 \draw (1.5,-1.5)node[bul](ll){} (1.5,1.5)node[bul](ul){}
 (3.5,-1.5)node[bul](lr){} (3.5,1.5)node[bul](ur){};
 \draw [->](s1) to (ul);\draw [->] (s1) to (ll); \draw [->](s2) to (ur);\draw [->] (s2) to (lr);
 \draw [->](ul) to (t2);\draw [->] (ll) to (t2); \draw [->](ur) to (t1);\draw [->] (lr) to (t1);
 \draw [red, line width=2pt] (ul)-- (ur) (ll)--(lr);
\end{scope}
\begin{scope}[shift={(-7,-11.5)},yscale=-1,xscale=1] 
\fill[orange!50] (-2,0)--(21.5,0)--(21.5,7)--(-2,7)--cycle;
         \draw (3,2)node[bul](t){} (10,2)node[bul](s){} (-1.5,0)node[bul](lx){} (.5,0)node[bul](rx){};
        \draw (4,1)node[bul,fill=blue!25](tt2){} (5,1)node[bul,fill=blue](ss1){} (8,1)node[bul,fill=blue!25](ss2){} (9,1)node[bul,fill=blue](tt1){}; 
        \draw (5.5,0)node[bul](ll){} (7.5,0)node[bul](lr){} (5.5,2)node[bul](ul){} (7.5,2)node[bul](ur){}; 
        \draw (19,0)node[bul](lz){} (21,0)node[bul](rz){}
        (3,3.5)node[bul](t1){} (10,3.5)node[bul](s1){};
        \draw (3,5.5)node[bul](t2){} (10,5.5)node[bul](s2){} (2,6.5)node[bul](t3){} (11,6.5)node[bul](s3){};
        \draw (6.5,4)node[bul](c1){} (6.5,6)node[bul](c2){};
        \draw (4.75,3.5)node[bul,fill=green!25](l3){} (4.75,4.5)node[bul,fill=green](ll4){} (4.75,5.3)node[bul,fill=green](l4){} (4.75,6)node[bul,fill=green!25](ll3){};
        \draw (8.25,3.5)node[bul,fill=red!25](r3){} (8.25,4.5)node[bul,fill=red](rr4){} (8.25,5.3)node[bul,fill=red](r4){} (8.25,6)node[bul,fill=red!25](rr3){};

        \draw [->] (t1)to(t);\draw [->] (t3)to(t2);\draw [->] (s)to(s1);\draw [->] (s2)to(s3);
        \draw [red,line width=2pt] (c1)--(c2) (t2)--(t1) (s2)--(s1);

        \draw [->] (rz)to(s3); \draw [->] (s3)to(t3); \draw [->] (t3)to(lx);
        \draw [->] (ll4)to(t1);\draw [->] (ll4)to(c1); \draw [->] (ll3)to(t2);\draw [->] (ll3)to(c2);
        \draw [->] (t1)to(l3); \draw [->] (c1)to(l3); \draw [->] (t2)to(l4);\draw [->] (c2)to(l4);

        \draw [->] (rr4)to(s1);\draw [->] (rr4)to(c1); \draw [->] (rr3)to(s2);\draw [->] (rr3)to(c2);
        \draw [->] (s1)to(r3); \draw [->] (c1)to(r3); \draw [->] (s2)to(r4);\draw [->] (c2)to(r4);
        \draw [->](ss1) to (ul);\draw [->] (ss1) to (ll);\draw [->](ss2) to (ur);\draw [->] (ss2) to (lr);
        \draw [->](ul) to (tt2);\draw [->] (ll) to (tt2);\draw [->](ur) to (tt1);\draw [->] (lr) to (tt1);
        \draw [->] (rx)to(t);\draw [->] (ul)to(t);\draw [->] (s)to(ur);\draw [->] (s)to(lz);

        \draw [red,line width=2pt] (ul)-- (ur) (ll)-- (lr);
        \draw [red,line width=2pt] (lx)-- (rx) (lz)-- (rz);
\end{scope}
\begin{scope}[shift={(13,-19)},yscale=-1,xscale=1] 
\fill[orange!50] (-28.5,0)--(15,0)--(15,7)--(-28.5,7)--cycle;
         \draw (3,2)node[bul](t){} (10,2)node[bul](s){} (-28,0)node[bul](lx){} (-26,0)node[bul](rx){};
        \draw (4,1)node[bul,fill=blue!25](tt2){} (5,1)node[bul,fill=blue](ss1){} (8,1)node[bul,fill=blue!25](ss2){} (9,1)node[bul,fill=blue](tt1){}; 
        \draw (5.5,0)node[bul](ll){} (7.5,0)node[bul](lr){} (5.5,2)node[bul](ul){} (7.5,2)node[bul](ur){}; 
        \draw (12.5,0)node[bul](lz){} (14.5,0)node[bul](rz){}
        (3,3.5)node[bul](t1){} (10,3.5)node[bul](s1){};
        \draw (3,5.5)node[bul](t2){} (10,5.5)node[bul](s2){} (2,6.5)node[bul](t3){} (11,6.5)node[bul](s3){};
        \draw (6.5,4)node[bul](c1){} (6.5,6)node[bul](c2){};
        \draw (4.75,3.5)node[bul,fill=green!25](l3){} (4.75,4.5)node[bul,fill=green](ll4){} (4.75,5.3)node[bul,fill=green](l4){} (4.75,6)node[bul,fill=green!25](ll3){};
        \draw (8.25,3.5)node[bul,fill=red!25](r3){} (8.25,4.5)node[bul,fill=red](rr4){} (8.25,5.3)node[bul,fill=red](r4){} (8.25,6)node[bul,fill=red!25](rr3){};

        \draw [->] (t1)to(t);\draw [->] (t3)to(t2);\draw [->] (s)to(s1);\draw [->] (s2)to(s3);
        \draw [red,line width=2pt] (c1)--(c2) (t2)--(t1) (s2)--(s1);

        \draw [->] (rz)to(s3); \draw [->] (s3)to(t3); \draw [->] (t3)to(lx);
        \draw [->] (ll4)to(t1);\draw [->] (ll4)to(c1); \draw [->] (ll3)to(t2);\draw [->] (ll3)to(c2);
        \draw [->] (t1)to(l3); \draw [->] (c1)to(l3); \draw [->] (t2)to(l4);\draw [->] (c2)to(l4);

        \draw [->] (rr4)to(s1);\draw [->] (rr4)to(c1); \draw [->] (rr3)to(s2);\draw [->] (rr3)to(c2);
        \draw [->] (s1)to(r3); \draw [->] (c1)to(r3); \draw [->] (s2)to(r4);\draw [->] (c2)to(r4);
        \draw [->](ss1) to (ul);\draw [->] (ss1) to (ll);\draw [->](ss2) to (ur);\draw [->] (ss2) to (lr);
        \draw [->](ul) to (tt2);\draw [->] (ll) to (tt2);\draw [->](ur) to (tt1);\draw [->] (lr) to (tt1);
        \draw [->] (rx)to(t);\draw [->] (ul)to(t);\draw [->] (s)to(ur);\draw [->] (s)to(lz);

        \draw [red,line width=2pt] (ul)-- (ur) (ll)-- (lr);
        \draw [red,line width=2pt] (lx)-- (rx) (lz)-- (rz);
\end{scope}
\begin{scope}[shift={(-10,-10)}]
 \draw (0,0)node[bul](lt2){} (1,0)node[bul](ls1){} (4,0)node[bul](ls2){} (5,0)node[bul](lt1){};
 \draw (1.5,-1.5)node[bul](ll){} (3.5,-1.5)node[bul](lr){}  (1.5,1.5)node[bul](cl){} (3.5,1.5)node[bul](cr){};
 \draw (1.5,4.5)node[bul](ul){} (3.5,4.5)node[bul](ur){};
 \draw (0,3)node[bul](ut2){} (1,3)node[bul](us1){} (4,3)node[bul](us2){} (5,3)node[bul](ut1){};
 \draw [->](ls1) to (cl);\draw [->] (ls1) to (ll); \draw [->](ls2) to (cr);\draw [->] (ls2) to (lr);
 \draw [->](cl) to (lt2);\draw [->] (ll) to (lt2); \draw [->](cr) to (lt1);\draw [->] (lr) to (lt1);
 \draw [->](us1) to (cl);\draw [->] (us1) to (ul); \draw [->](us2) to (cr);\draw [->] (us2) to (ur);
 \draw [->](ul) to (ut2);\draw [->] (cl) to (ut2); \draw [->](ur) to (ut1);\draw [->] (cr) to (ut1);
 \draw [red,line width=2pt] (ul) -- (ur) (cl) -- (cr) (ll) -- (lr);
\end{scope}
\begin{scope}[shift={(-16.5,-10)}]
 \draw (0,0)node[bul](lt2){} (1,0)node[bul](ls1){} (4,0)node[bul](ls2){} (5,0)node[bul](lt1){};
 \draw (1.5,-9)node[bul](ll){} (3.5,-9)node[bul](lr){}  (1.5,1.5)node[bul](cl){} (3.5,1.5)node[bul](cr){};
 \draw (1.5,4.5)node[bul](ul){} (3.5,4.5)node[bul](ur){};
 \draw (0,3)node[bul](ut2){} (1,3)node[bul](us1){} (4,3)node[bul](us2){} (5,3)node[bul](ut1){};
 \draw [->](ls1) to (cl);\draw [->] (ls1) to (ll); \draw [->](ls2) to (cr);\draw [->] (ls2) to (lr);
 \draw [->](cl) to (lt2);\draw [->] (ll) to (lt2); \draw [->](cr) to (lt1);\draw [->] (lr) to (lt1);
 \draw [->](us1) to (cl);\draw [->] (us1) to (ul); \draw [->](us2) to (cr);\draw [->] (us2) to (ur);
 \draw [->](ul) to (ut2);\draw [->] (cl) to (ut2); \draw [->](ur) to (ut1);\draw [->] (cr) to (ut1);
 \draw [red,line width=2pt] (ul) -- (ur) (cl) -- (cr) (ll) -- (lr);
\end{scope}
\begin{scope}[shift={(-16.5,-1)}]
 \draw (0,0)node[bul](lt2){} (1,0)node[bul](ls1){} (4,0)node[bul](ls2){} (5,0)node[bul](lt1){};
 \draw (1.5,-1.5)node[bul](ll){} (3.5,-1.5)node[bul](lr){}  (1.5,1.5)node[bul](cl){} (3.5,1.5)node[bul](cr){};
 \draw (1.5,4.5)node[bul](ul){} (3.5,4.5)node[bul](ur){};
 \draw (0,3)node[bul](ut2){} (1,3)node[bul](us1){} (4,3)node[bul](us2){} (5,3)node[bul](ut1){};
 \draw [->](ls1) to (cl);\draw [->] (ls1) to (ll); \draw [->](ls2) to (cr);\draw [->] (ls2) to (lr);
 \draw [->](cl) to (lt2);\draw [->] (ll) to (lt2); \draw [->](cr) to (lt1);\draw [->] (lr) to (lt1);
 \draw [->](us1) to (cl);\draw [->] (us1) to (ul); \draw [->](us2) to (cr);\draw [->] (us2) to (ur);
 \draw [->](ul) to (ut2);\draw [->] (cl) to (ut2); \draw [->](ur) to (ut1);\draw [->] (cr) to (ut1);
 \draw [red,line width=2pt] (ul) -- (ur) (cl) -- (cr) (ll) -- (lr);
\end{scope}
\begin{scope}[shift={(-3,-1)}]
 \draw (0,0)node[bul](lt2){} (1,0)node[bul](ls1){} (4,0)node[bul](ls2){} (5,0)node[bul](lt1){};
 \draw (1.5,-1.5)node[bul](ll){} (3.5,-1.5)node[bul](lr){}  (1.5,1.5)node[bul](cl){} (3.5,1.5)node[bul](cr){};
 \draw (1.5,4.5)node[bul](ul){} (3.5,4.5)node[bul](ur){};
 \draw (0,3)node[bul](ut2){} (1,3)node[bul](us1){} (4,3)node[bul](us2){} (5,3)node[bul](ut1){};
 \draw [->](ls1) to (cl);\draw [->] (ls1) to (ll); \draw [->](ls2) to (cr);\draw [->] (ls2) to (lr);
 \draw [->](cl) to (lt2);\draw [->] (ll) to (lt2); \draw [->](cr) to (lt1);\draw [->] (lr) to (lt1);
 \draw [->](us1) to (cl);\draw [->] (us1) to (ul); \draw [->](us2) to (cr);\draw [->] (us2) to (ur);
 \draw [->](ul) to (ut2);\draw [->] (cl) to (ut2); \draw [->](ur) to (ut1);\draw [->] (cr) to (ut1);
 \draw [red,line width=2pt] (ul) -- (ur) (cl) -- (cr) (ll) -- (lr);
\end{scope}
\begin{scope}[shift={(-3,-10)}]
 \draw (0,0)node[bul](lt2){} (1,0)node[bul](ls1){} (4,0)node[bul](ls2){} (5,0)node[bul](lt1){};
 \draw (1.5,-1.5)node[bul](ll){} (3.5,-1.5)node[bul](lr){}  (1.5,1.5)node[bul](cl){} (3.5,1.5)node[bul](cr){};
 \draw (1.5,4.5)node[bul](ul){} (3.5,4.5)node[bul](ur){};
 \draw (0,3)node[bul](ut2){} (1,3)node[bul](us1){} (4,3)node[bul](us2){} (5,3)node[bul](ut1){};
 \draw [->](ls1) to (cl);\draw [->] (ls1) to (ll); \draw [->](ls2) to (cr);\draw [->] (ls2) to (lr);
 \draw [->](cl) to (lt2);\draw [->] (ll) to (lt2); \draw [->](cr) to (lt1);\draw [->] (lr) to (lt1);
 \draw [->](us1) to (cl);\draw [->] (us1) to (ul); \draw [->](us2) to (cr);\draw [->] (us2) to (ur);
 \draw [->](ul) to (ut2);\draw [->] (cl) to (ut2); \draw [->](ur) to (ut1);\draw [->] (cr) to (ut1);
 \draw [red,line width=2pt] (ul) -- (ur) (cl) -- (cr) (ll) -- (lr);
\end{scope}
\begin{scope}[shift={(3.5,-1)}]
 \draw (0,0)node[bul](lt2){} (1,0)node[bul](ls1){} (4,0)node[bul](ls2){} (5,0)node[bul](lt1){};
 \draw (1.5,-1.5)node[bul](ll){} (3.5,-1.5)node[bul](lr){}  (1.5,1.5)node[bul](cl){} (3.5,1.5)node[bul](cr){};
 \draw (1.5,4.5)node[bul](ul){} (3.5,4.5)node[bul](ur){};
 \draw (0,3)node[bul](ut2){} (1,3)node[bul](us1){} (4,3)node[bul](us2){} (5,3)node[bul](ut1){};
 \draw [->](ls1) to (cl);\draw [->] (ls1) to (ll); \draw [->](ls2) to (cr);\draw [->] (ls2) to (lr);
 \draw [->](cl) to (lt2);\draw [->] (ll) to (lt2); \draw [->](cr) to (lt1);\draw [->] (lr) to (lt1);
 \draw [->](us1) to (cl);\draw [->] (us1) to (ul); \draw [->](us2) to (cr);\draw [->] (us2) to (ur);
 \draw [->](ul) to (ut2);\draw [->] (cl) to (ut2); \draw [->](ur) to (ut1);\draw [->] (cr) to (ut1);
 \draw [red,line width=2pt] (ul) -- (ur) (cl) -- (cr) (ll) -- (lr);
\end{scope}
\begin{scope}[shift={(10.5,-1)}]
 \draw (0,0)node[bul](lt2){} (1,0)node[bul](ls1){} (4,0)node[bul](ls2){} (5,0)node[bul](lt1){};
 \draw (1.5,-1.5)node[bul](ll){} (3.5,-1.5)node[bul](lr){}  (1.5,1.5)node[bul](cl){} (3.5,1.5)node[bul](cr){};
 \draw (1.5,4.5)node[bul](ul){} (3.5,4.5)node[bul](ur){};
 \draw (0,3)node[bul](ut2){} (1,3)node[bul](us1){} (4,3)node[bul](us2){} (5,3)node[bul](ut1){};
 \draw [->](ls1) to (cl);\draw [->] (ls1) to (ll); \draw [->](ls2) to (cr);\draw [->] (ls2) to (lr);
 \draw [->](cl) to (lt2);\draw [->] (ll) to (lt2); \draw [->](cr) to (lt1);\draw [->] (lr) to (lt1);
 \draw [->](us1) to (cl);\draw [->] (us1) to (ul); \draw [->](us2) to (cr);\draw [->] (us2) to (ur);
 \draw [->](ul) to (ut2);\draw [->] (cl) to (ut2); \draw [->](ur) to (ut1);\draw [->] (cr) to (ut1);
 \draw [red,line width=2pt] (ul) -- (ur) (cl) -- (cr) (ll) -- (lr);
\end{scope}
\begin{scope}[shift={(10.5,-10)}]
 \draw (0,0)node[bul](lt2){} (1,0)node[bul](ls1){} (4,0)node[bul](ls2){} (5,0)node[bul](lt1){};
 \draw (1.5,-1.5)node[bul](ll){} (3.5,-1.5)node[bul](lr){}  (1.5,1.5)node[bul](cl){} (3.5,1.5)node[bul](cr){};
 \draw (1.5,4.5)node[bul](ul){} (3.5,4.5)node[bul](ur){};
 \draw (0,3)node[bul](ut2){} (1,3)node[bul](us1){} (4,3)node[bul](us2){} (5,3)node[bul](ut1){};
 \draw [->](ls1) to (cl);\draw [->] (ls1) to (ll); \draw [->](ls2) to (cr);\draw [->] (ls2) to (lr);
 \draw [->](cl) to (lt2);\draw [->] (ll) to (lt2); \draw [->](cr) to (lt1);\draw [->] (lr) to (lt1);
 \draw [->](us1) to (cl);\draw [->] (us1) to (ul); \draw [->](us2) to (cr);\draw [->] (us2) to (ur);
 \draw [->](ul) to (ut2);\draw [->] (cl) to (ut2); \draw [->](ur) to (ut1);\draw [->] (cr) to (ut1);
 \draw [red,line width=2pt] (ul) -- (ur) (cl) -- (cr) (ll) -- (lr);
\end{scope}
\begin{scope}[shift={(17,-10)}]
 \draw (0,0)node[bul](lt2){} (1,0)node[bul](ls1){} (4,0)node[bul](ls2){} (5,0)node[bul](lt1){};
 \draw (1.5,-9)node[bul](ll){} (3.5,-9)node[bul](lr){}  (1.5,1.5)node[bul](cl){} (3.5,1.5)node[bul](cr){};
 \draw (1.5,4.5)node[bul](ul){} (3.5,4.5)node[bul](ur){};
 \draw (0,3)node[bul](ut2){} (1,3)node[bul](us1){} (4,3)node[bul](us2){} (5,3)node[bul](ut1){};
 \draw [->](ls1) to (cl);\draw [->] (ls1) to (ll); \draw [->](ls2) to (cr);\draw [->] (ls2) to (lr);
 \draw [->](cl) to (lt2);\draw [->] (ll) to (lt2); \draw [->](cr) to (lt1);\draw [->] (lr) to (lt1);
 \draw [->](us1) to (cl);\draw [->] (us1) to (ul); \draw [->](us2) to (cr);\draw [->] (us2) to (ur);
 \draw [->](ul) to (ut2);\draw [->] (cl) to (ut2); \draw [->](ur) to (ut1);\draw [->] (cr) to (ut1);
 \draw [red,line width=2pt] (ul) -- (ur) (cl) -- (cr) (ll) -- (lr);
\end{scope}
\begin{scope}[shift={(24,-10)}]
 \draw (0,0)node[bul](lt2){} (1,0)node[bul](ls1){} (4,0)node[bul](ls2){} (5,0)node[bul](lt1){};
 \draw (1.5,-9)node[bul](ll){} (3.5,-9)node[bul](lr){}  (1.5,1.5)node[bul](cl){} (3.5,1.5)node[bul](cr){};
 \draw (1.5,4.5)node[bul](ul){} (3.5,4.5)node[bul](ur){};
 \draw (0,3)node[bul](ut2){} (1,3)node[bul](us1){} (4,3)node[bul](us2){} (5,3)node[bul](ut1){};
 \draw [->](ls1) to (cl);\draw [->] (ls1) to (ll); \draw [->](ls2) to (cr);\draw [->] (ls2) to (lr);
 \draw [->](cl) to (lt2);\draw [->] (ll) to (lt2); \draw [->](cr) to (lt1);\draw [->] (lr) to (lt1);
 \draw [->](us1) to (cl);\draw [->] (us1) to (ul); \draw [->](us2) to (cr);\draw [->] (us2) to (ur);
 \draw [->](ul) to (ut2);\draw [->] (cl) to (ut2); \draw [->](ur) to (ut1);\draw [->] (cr) to (ut1);
 \draw [red,line width=2pt] (ul) -- (ur) (cl) -- (cr) (ll) -- (lr);
\end{scope}

\begin{scope}[shift={(24,-1)}]
 \draw (0,0)node[bul](lt2){} (1,0)node[bul](ls1){} (4,0)node[bul](ls2){} (5,0)node[bul](lt1){};
 \draw (1.5,-1.5)node[bul](ll){} (3.5,-1.5)node[bul](lr){}  (1.5,1.5)node[bul](cl){} (3.5,1.5)node[bul](cr){};
 \draw (1.5,4.5)node[bul](ul){} (3.5,4.5)node[bul](ur){};
 \draw (0,3)node[bul](ut2){} (1,3)node[bul](us1){} (4,3)node[bul](us2){} (5,3)node[bul](ut1){};
 \draw [->](ls1) to (cl);\draw [->] (ls1) to (ll); \draw [->](ls2) to (cr);\draw [->] (ls2) to (lr);
 \draw [->](cl) to (lt2);\draw [->] (ll) to (lt2); \draw [->](cr) to (lt1);\draw [->] (lr) to (lt1);
 \draw [->](us1) to (cl);\draw [->] (us1) to (ul); \draw [->](us2) to (cr);\draw [->] (us2) to (ur);
 \draw [->](ul) to (ut2);\draw [->] (cl) to (ut2); \draw [->](ur) to (ut1);\draw [->] (cr) to (ut1);
 \draw [red,line width=2pt] (ul) -- (ur) (cl) -- (cr) (ll) -- (lr);
\end{scope}
\end{tikzpicture}
\clearpage{}
	  \caption{Corresponding \textsc{Planar Steiner Orientation} instance}
		\label{fig:full-example-orientation-instance}
    \end{subfigure}

	\caption{A full example showing the reduction from \textsc{Planar Monotone 3-SAT} to \textsc{Planar Steiner Orientation}. Variables and variable gadgets are highlighted in green, clauses and clause gadgets in orange. Negative clause gadgets are mirrored vertically.}
	\label{fig:full-example}
\end{figure}
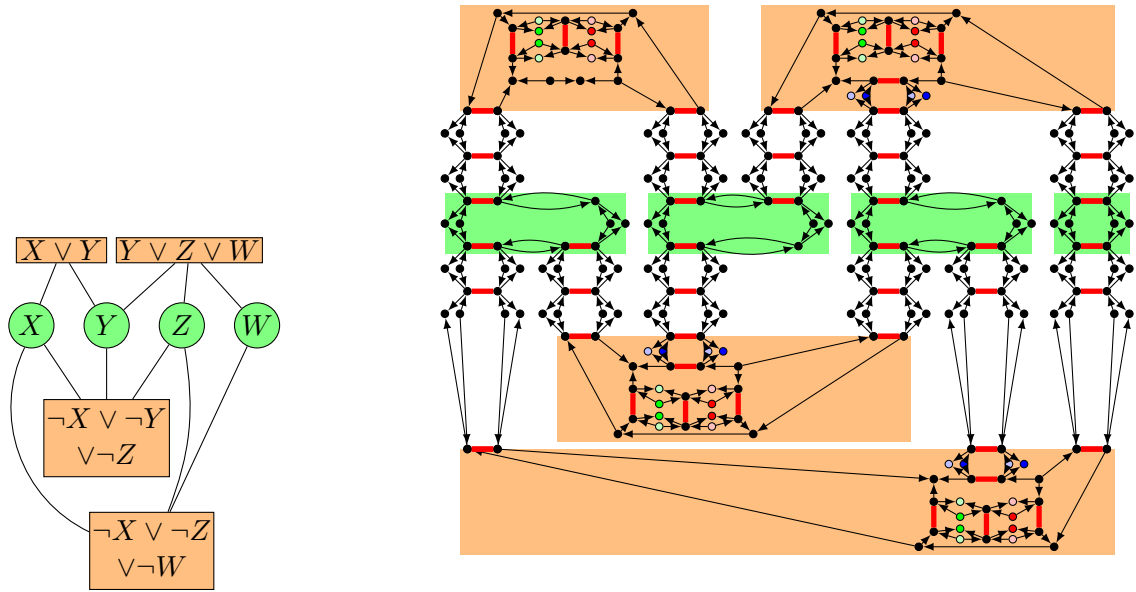

\end{document}